\documentclass[a4paper,UKenglish,cleveref, autoref, thm-restate]{lipics-v2021}

\pdfoutput=1 
\hideLIPIcs  

\title{Order-Invariance in the Two-Variable Fragment of First-Order Logic}

\author{Julien Grange}{LACL, Université Paris-Est Créteil, France}{julien.grange@lacl.fr}{}{}

\authorrunning{J. Grange} 

\Copyright{Julien Grange} 

\ccsdesc[100]{Theory of computation~Finite Model Theory}

\keywords{Finite model theory, two-variable logic, order-invariance} 

\category{} 

\relatedversion{} 

\nolinenumbers 

\usepackage{macros_csl23}

\NewDocumentCommand{\labsr}{O{r}}{S_{#1}}
\NewDocumentCommand{\refsr}{O{r}}{$(\labsr[#1])$\xspace}

\NewDocumentCommand{\laber}{O{r}}{E_{#1}}
\NewDocumentCommand{\refer}{O{r}}{$(\laber[#1])$\xspace}

\NewDocumentCommand{\labtr}{O{r}}{T_{#1}}
\NewDocumentCommand{\reftr}{O{r}}{$(\labtr[#1])$\xspace}

\EventEditors{John Q. Open and Joan R. Access}
\EventNoEds{2}
\EventLongTitle{42nd Conference on Very Important Topics (CVIT 2016)}
\EventShortTitle{CVIT 2016}
\EventAcronym{CVIT}
\EventYear{2016}
\EventDate{December 24--27, 2016}
\EventLocation{Little Whinging, United Kingdom}
\EventLogo{}
\SeriesVolume{42}
\ArticleNo{23}

\begin{document}

\maketitle

\begin{abstract}
  We study the expressive power of the two-variable fragment of order-invariant first-order logic. This logic departs from first-order logic in two ways: first, formulas are only allowed to quantify over two variables. Second, formulas can use an additional binary relation, which is interpreted in the structures under scrutiny as a linear order, provided that the truth value of a sentence over a finite structure never depends on which linear order is chosen on its domain.

  We prove that on classes of structures of bounded degree, any property expressible in this logic is definable in first-order logic.
  We then show that the situation remains the same when we add counting quantifiers to this logic.
\end{abstract}

\section{Introduction}

The restriction of first-order logic to two variables (\FOtwo) holds an important place among the fragments of first-order logic (\FO), since it is the maximal fragment, with respect to the number of variables, for which the finite and general satisfiability problems are decidable~\cite{DBLP:journals/mlq/Mortimer75} - see~\cite{DBLP:journals/tcs/GradelO99} for a complete survey on these issues. They become undecidable as soon as we allow formulas to make use of three variables, as three variables are enough to encode grids. Tame as it is in this regard, \FOtwo is still fairly expressive (in particular, it embeds modal logic). It is thus natural to investigate its order-invariant extension \oifotwo.

In the order-invariant extension $\oifol$ of a logic \logic, one can make use in the \logic-formulas of a linear order on the vertices of the structures at hand, provided that the validity of said formulas in a given finite structure does not depend on the choice of a particular order.

If $\logic=\FO$, we get \oifo, whose syntax is not recursively enumerable. It has however been proven by Harwath and Zeume~\cite{DBLP:conf/lics/ZeumeH16} that on the other hand, \oifotwo has a recursive syntax. The natural follow-up to this result is to study the expressive power of \oifotwo. This shall be our endeavor in this article.

Perhaps surprisingly, it has been shown by Gurevich (see Section 5.2 of~\cite{DBLP:books/sp/Libkin04}) that such an order can indeed bring additional expressive power to \FO, even when restricted in this way: there exist properties which are not definable in \FO, but which can be expressed as soon as the use of a linear order is authorized, even in an invariant fashion. In symbols: $\oifo\not\subseteq\FO$. It is not hard to observe that $\oifotwo\not\subseteq\FOtwo$ (for instance, one can state in \oifotwo that a set has at least three elements, which is not possible in \FOtwo). It is however not clear whether $\oifotwo\subseteq\FO$, or whether even when restricted to two variables, the addition of an order allows one to express properties beyond the scope of \FO. This question is asked in~\cite{DBLP:conf/lics/ZeumeH16}. Is this paper, we prove that on any class of bounded degree, the inclusion $\oifotwo\subseteq\FO$ indeed holds - this is Theorem~\ref{th:main}, whose proof is the object of Sections~\ref{sec:construction} and \ref{sec:EFproof}.

We then explain in Section~\ref{sec:counting} how to extend this result from \oifotwo to \oictwo, where \Ctwo is the extension of \FOtwo with counting quantifiers. Precisely, we show with Theorem~\ref{th:mainCtwo} that $\oictwo\subseteq\FO$ when the degree is bounded.

Let us already mention that both of these inclusions are strict. This matter is further discussed in the conclusion.

\paragraph*{Related work} The question of the order-invariance of an \FOtwo-sentence has been shown to be decidable in~\cite{DBLP:conf/lics/ZeumeH16}. The expressive power of \oifotwo is only mentioned there as a follow up question.

We borrow the dichotomy between rare and frequent neighborhood types when the degree is bounded from \cite{DBLP:conf/lics/Grange20}. Beyond that, the philosophies of the constructions differ widely: in \cite{DBLP:conf/lics/Grange20}, successor relations are constructed in a very regular way, in order to create as few neighborhood types as possible in the structures with successor. On the other hand, we make sure to realize all the possible types in our construction.

One line of research investigates the expressive power of \oifo. Let us mention~\cite{DBLP:journals/jsyml/BenediktS09} and \cite{grange_et_al:LIPIcs:2020:11666}, which prove that \oifo has the same expressive power as \FO respectively on trees and on hollow trees. The present paper focuses on a weaker logic, but in a broader setting. Furthermore, while the techniques used in these papers involve the construction of several intermediate orders, making only localized changes at each step (in the fashion of \cite{DBLP:journals/tocl/GroheS00}, in which it is proved that \oifo retains the locality of \FO), we equip in one go each of our structures with a linear order.

Although not directly related to our inquiries, \cite{DBLP:journals/corr/abs-0907-0616} is also concerned with the expressive power of \FOtwo on ordered structures. This paper establishes a strict hierarchy, based on the quantifier rank and quantifier alternation, among properties definable in \FOtwo on words.

\section{Preliminaries}
\label{sec:prelim}

\subsection{General definitions}

We use the standard definition of first-order logic $\FO(\Sigma)$ with equality (written \FO when $\vocab$ is clear from the context) on a finite signature $\vocab$ composed of relation and constant symbols. By \FOtwo we denote the fragment of \FO in which the only two variables are $\var$ and $\varbis$.

Structures are denoted by a name that starts with (or consists of) a calligraphic upper-case letter, while their universes are denoted by the same name starting with a standard upper-case letter instead of the calligraphic one; for instance, $E\!x$ is the universe of the structure $\mathcal E\!x$. Throughout this paper, we consider only finite structures.

A sentence $\formule\in\FOtwo(\vocab\cup\{<\})$, where $<$ is a binary
relation symbol not belonging to $\vocab$, is said to be \textbf{order-invariant} if for every finite
$\vocab$-structure $\structgen$, and every pair of strict linear orders $\ord$ and
$\ordbis$ on $\structgendom$, $(\structgen,\ord)\models\formule$ iff $(\structgen,\ordbis)\models\formule$.
It is then convenient to omit the interpretation for
the symbol $<$, and to write $\structgen\models\formule$ iff $(\structgen,\ord)\models\formule$ for any (or, equivalently, every)
linear order $\ord$.

The set of order-invariant sentences using two variables is denoted $\oifotwo$.

Let $\logic,\logicbis$ be two logics defined over the same signature, and $\classe$ be a class of finite structures on this signature.

We say that a property $\prop\subseteq\classe$ is \textbf{definable} (or expressible) in $\logic$ if there exists an $\logic$-sentence $\formule$ such that $\prop=\{\structgen\in\classe:\ \structgen\models\formule\}$.

We say that $\logic\subseteq\logicbis$ on $\classe$ if every property on $\classe$ definable in $\logic$ is also definable in $\logicbis$.

It is quite clear that $\FOtwo\subseteq\ \oifotwo$: any sentence which does not make use of the order is indeed order-invariant. Furthermore, this inclusion is strict. For instance, over the empty signature, the property of having at least three elements is not definable in \FOtwo (this can easily be seen with the tools presented in Section~\ref{sec:EF}), but can be expressed in \oifotwo, for instance via the formula
$\exists\var\ \exists\varbis\ (\var < \varbis\ \land\ (\exists\var\ \varbis<\var))\,.$

The \textbf{quantifier rank} of a formula is the maximal number of quantifiers in a branch of its syntactic tree.
Given two $\vocab$-structures $\struct$ and $\structbis$, and $\logic$ being one of $\FO,\FOtwo$ and $\oifotwo$, we write $\struct\equiv^{\mathcal L}_k\structbis$ if $\struct$ and $\structbis$ satisfy the same $\logic$-sentences of quantifier rank at most $k$. In this case, we say that $\struct$ and $\structbis$ are \textbf{$\logic$-similar at depth $k$}.

We write $\struct\simeq\structbis$ if $\struct$ and $\structbis$ are isomorphic.

\paragraph*{Atomic types}
Let $\elem$ be an element of a structure $\structgen$. The \textbf{atomic type} $\tpgen$ of $\elem$ in $\structgen$ is the set of atomic formulas $\formule$ with at most one free variable $\var$ such that $\structgen,\var\mapsto\elem\models\formule$.

We define similarly the atomic type $\tpgen[\elem,\elembis]$ of a pair $(\elem,\elembis)$ of elements of $\structgen$ as the set of atomic formulas $\formule$ with free variables in $\{\var,\varbis\}$ such that $\structgen,\var\mapsto\elem,\varbis\mapsto\elembis\models\formule$.

Given a  linearly ordered $\vocab$-structure $\orderstructgen$, $\tpgen[\elem,\elembis][\orderstructgen]$ can be split into $\ordtpgen$ and $\tpgen[\elem,\elembis]$, where $\ordtpgen$ is one of $\{x<y\},\{x>y\}$ and $\{x=y\}$.

\paragraph*{Gaifman graphs}

The \textbf{Gaifman graph} $\gaifman$ of a structure
$\structgen$ is defined as $(\structgendom,\edgerel)$ where $(\elem,\elembis)\in\edgerel$ iff $\elem$ and $\elembis$
appear in the same tuple of a relation of $\structgen$. Notice that if a graph is seen as a structure on the signature consisting of a single binary relation symbol, its
Gaifman graph is none other than the unoriented version of itself.

By $\dist{\structgen}{\elem}{\elembis}$, we denote the distance between $\elem$ and $\elembis$ in
$\gaifman$.

For $B\subseteq\structgendom$, we note $\Neighgen{B}$ the set of elements at distance exactly $1$ from $B$ in $\gaifman$. In particular, $B\cap\Neighgen{B}=\emptyset$.

The \textbf{degree} of $\structgen$ is the
maximal degree of its Gaifman graph, and a class $\classe$ of
$\vocab$-structures is said to have \textbf{bounded degree} if there
exists some $d\in\N$ such that the degree of every
$\structgen\in\classe$ is at most $d$.

\subsection{Main result}
\label{sec:main}

We are now able to state the main result of this article. Remember that \oifotwo allows us to express properties that are beyond the scope of plain \FOtwo. We give an upper bound to its expressive power, when the degree is bounded:

\begin{theorem}
  \label{th:main}
  Let $\classe$ be a class of structures of bounded degree.
\\  Then $\oifotwo\subseteq\FO$ on $\classe$.
\end{theorem}

For the remainder of this paper, we fix a signature $\vocab$, an integer $d$ and a class $\classe$ of $\vocab$-structures of degree at most $d$.

Let us now show the skeleton of our proof. The technical part of the proof will be the focus of Sections~\ref{sec:construction} and~\ref{sec:EFproof}.

Our general strategy is to show the existence of a function $f:\N\to\N$ such that every formula $\formule\in\ \oifotwo$ of quantifier rank $k$ is equivalent on $\classe$ (i.e. satisfied by the same structures of $\classe$) to an \FO-formula $\formulebis$ of quantifier rank at most $f(k)$.

To prove this, it is enough to show that for any two structures $\struct,\structbis\in\classe$ such that $\struct\foeq{f(k)}\structbis$, we have $\struct\oifotwoeq{k}\structbis$. Indeed, the class of structures satisfying a formula $\formule\in\ \oifotwo$ of quantifier rank $k$ is a union of equivalence classes for the equivalence relation $\oifotwoeq{k}$, whose intersection with $\classe$ is in turn the intersection of $\classe$ with a union of equivalence classes for $\foeq{f(k)}$. It is folklore (see, e.g., \cite{DBLP:books/sp/Libkin04}) that this relation has finite index, and that each of its equivalence classes is definable by an \FO-sentence of quantifier rank $f(k)$. Then $\formulebis$ is just the finite disjunction of these \FO-sentences.

In order to show that $\struct\oifotwoeq{k}\structbis$, we will construct in Section~\ref{sec:construction} two particular orders $\ord,\ordbis$ on these respective structures, and we will prove in Section~\ref{sec:EFproof} that \begin{equation}\orderstruct\fotwoeq{k}\orderstructbis\,.\label{eq:order_eq}\end{equation}
This concludes the proof, since any sentence $\formuleter\in\ \oifotwo$ with quantifier rank at most $k$ holds in $\struct$ iff it holds in $\orderstruct$ (by definition of order-invariance), iff it holds in $\orderstructbis$ (by (\ref{eq:order_eq})), iff it holds in $\structbis$.

\section{\texorpdfstring{Constructing linear orders on $\struct$ and $\structbis$}{Constructing linear orders on A0 and A1}}
\label{sec:construction}

Recall from Section~\ref{sec:main} that our goal is to find a function $f$ such that, given two structures $\struct,\structbis$ in $\classe$ such that
\begin{equation}
  \label{eq:foeq}
  \struct\foeq{f(k)}\structbis\,,
\end{equation}
we are able to construct two linear orders $\ord,\ordbis$ such that $\orderstruct\fotwoeq{k}\orderstructbis$.

In this section, we set $f$ and we detail the construction of such orders. The proof of \oifo-similarity between $\orderstruct$ and $\orderstructbis$ will be the focus of Section~\ref{sec:EFproof}.

Let us now explain how we define $f$. For that, we need to introduce the notion of neighborhood and neighborhood type. These notions are defined in Section~\ref{sec:defneigh}. We then explain in Section~\ref{sec:freqneigh} how to divide neighborhood types into rare ones and frequent ones. Finally, the details of the construction are given in Section~\ref{sec:orders}.

\subsection{Neighborhoods}
\label{sec:defneigh}

Let us now define the notion of neighborhood of an element in a structure.

Let $c$ be a new constant symbol, and let $\structgen\in\classe$.

For $k\in\N$ and $\elem\in\structgendom$, the (pointed) \textbf{$k$-neighborhood}
$\neighgen$ of $\elem$ in $\structgen$ is the $(\vocab\cup\{c\})$-structure
whose restriction to the vocabulary $\vocab$ is the substructure of $\structgen$ induced by the set
$\boulegen=\{\elembis\in\structgendom:\dist{\structgen}{\elem}{\elembis}\leq k\}\,,$ and where $c$ is interpreted
as $\elem$. In other words, it consists of all the elements at distance at most $k$ from $\elem$ in $\structgen$, together with the relations they share in $\structgen$; the center $\elem$ being marked by the constant $c$. We sometimes refer to $\boulegen$ as the $k$-neighborhood of $a$ in $\structgen$ as well, but the context will always make clear whether we refer to the whole substructure or only its domain.

The \textbf{$k$-neighborhood type} $\type=\ntpgen$ of $\elem$ in $\structgen$ is the
isomorphism class of its $k$-neighborhood. We say that $\type$ is a $k$-neighborhood type
over $\vocab$, and that $\elem$ is an \textbf{occurrence} of
$\type$. We denote by $\nocc{\structgen}{\type}$ the number of occurrences of
$\type$ in $\structgen$, and we write
$\threq{\struct}{\structbis}{k}{t}$ to mean that for every $k$-neighborhood type
$\type$, $\nocc{\struct}{\type}$ and $\nocc{\structbis}{\type}$ are
either equal, or both larger than $t$.

Let $\Ntp$ denote the set of $k$-neighborhood types over $\vocab$ occurring in structures of degree at most $d$. Note that $\Ntp$ is a finite set.

The interest of this notion resides in the fact that when the degree is bounded, \FO is exactly able to count the number of occurrences of neighborhood types up to some threshold~\cite{DBLP:journals/iandc/FaginSV95}. We will only use one direction of this characterization, namely:

\begin{proposition}
  \label{prop:FO_neigh}
  For all integers $k$ and $t$, there exists some $\hat f(k,t)\in\N$ (which also depends on the bound $d$ on the degree of structures in $\classe$) such that for all structures $\struct,\structbis\in\classe$,
  \[\struct\foeq{\hat f(k,t)}\structbis\quad\to\quad\threq{\struct}{\structbis}{k}{t}\,.\]  
\end{proposition}

We now exhibit a function $\thr:\N\to\N$ such that, if $\threq{\struct}{\structbis}{k}{\thr(k)}$, then one can construct $\ord,\ordbis$ satisfying~(\ref{eq:order_eq}). Proposition~\ref{prop:FO_neigh} then ensures that $f:k\mapsto\hat f(k,\thr(k))$ fits the bill.

Let us now explain how the function $\thr$ is chosen.

\subsection{Frequency of a neighborhood type}
\label{sec:freqneigh}

Let us denote $|\Ntp|$ as $N$.

Recall that every $\structgen\in\classe$ has degree at most $d$. What this means is that if we consider the $k$-neighborhood types that have enough occurrences in $\structgen$, they must each have many occurrences that are scattered across $\structgen$. Not only that, but we can also make sure that such occurrences are far from all the occurrences of every other $k$-neighborhood type (which have few occurrences in $\structgen$, by definition). Since the degree is bounded, $N$ is bounded too, which prevents our distinction (which will be formalized later on) between rare neighborhood types and frequent neighborhood types from being circular.

Such a dichotomy is introduced and detailed in~\cite{DBLP:conf/lics/Grange20}; we simply adapt this construction to our needs. In the remainder of this section, we describe this construction at a high level, and leave the technical details (such as the exact bounds) to the reader.

The proof of the following lemma (in the vein of \cite{atserias2008preservation}) is straightforward, and relies on the degree boundedness hypothesis. Intuitively, Lemma~\ref{lem:scatter} states that when the degree is bounded, it is not possible for all the elements of large sets to be concentrated in one corner of the structure, thus making it possible to pick elements in each set that are scattered across the structure.

\begin{lemma}
  \label{lem:scatter}
  Given three integers $m$, $\delta$, $s$, there exists a threshold $g(m,\delta,s)\in\N$ such that for all $\structgen\in\classe$, all $B\subseteq\structgendom$ of size at most $s$, and all subsets $C_1,\cdots,C_n\subseteq\structgendom$ (with $n\leq N$) of size at least $g(m,\delta,s)$, it is possible to find elements $c_j^1,\cdots,c_j^m\in C_j$ for all $j\in\{1,\cdots,n\}$, such that for all $j,j'\in\{1,\cdots,n\}$ and $i,i'\in\{1,\cdots,m\}$, $\dist{\structgen}{c_j^i}{B}>\delta$ and $\dist{\structgen}{c_j^i}{c_{j'}^{i'}}>\delta$ if $(j,i)\neq(j',i')$.
\end{lemma}

Our goal is, given a structure $\structgen\in\classe$, to partition the $k$-neighborhood types into two classes: the frequent types, and the rare types. The property we wish to ensure is that there exist in $\structgen$ some number $m$ (which will be made precise later on, but only depends on $k$) of occurrences of each one of the frequent $k$-neighborhood types which are both
\begin{itemize}
\item at distance greater than $\delta$ (which, as for $m$, is a function of $k$ and will be fixed in the following) from one another, and
\item at distance greater than $\delta$ from every occurrence of a rare $k$-neighborhood type.
\end{itemize}

To establish this property, we would like to use Lemma~\ref{lem:scatter}, with $s$ being the total number of occurrences of all the rare $k$-neighborhood types, and $C_1,\cdots,C_n$ being the sets of occurrences of the $n$ distinct frequent $k$-neighborhood types.

The number $N$ of different $k$-neighborhood types of degree at most $d$ is bounded by a function of $k$ (as well as $\Sigma$ and $d$, which are fixed). Hence, we can proceed according to the following (terminating) algorithm to make the distinction between frequent and rare types:
\begin{enumerate}
\item First, let us mark every $k$-neighborhood type as frequent.
\item\label{enu:while} Among the types which are currently marked as frequent, let $\tau$ be one with the smallest number of occurrences in $\structgen$.
\item If $\nocc{\structgen}{\tau}$ is at least $g(m,\delta,s)$ ($g$ being the function from Lemma~\ref{lem:scatter}) where $s$ is the total number of occurrences of all the $k$-neighborhood types which are currently marked as rare, then we are done and the marking frequent/rare is final. Otherwise, mark $\tau$ as rare, and go back to step~\ref{enu:while} if there remains at least one frequent $k$-neighborhood type.
\end{enumerate}

Notice that we can go at most $N$ times through step~\ref{enu:while} - $N$ depending only on $k$. Furthermore, each time we add a type to the set of rare $k$-neighborhood types, we have the guarantee that this type has few occurrences (namely, less than $g(m,\delta,s)$, where $s$ can be bounded by a function of $k$).

It is thus apparent that the threshold $t$ such that a $k$-neighborhood type $\tau$ is frequent in $\structgen$ iff $\nocc{\structgen}{\tau}\geq t$ can be bounded by some $T$ depending only on $k$ - importantly, $T$ is the same for all structures of $\classe$.

Let us now make the above more formal. For $t\in\N$ and $\structgen\in\classe$, let $\Freq[\structgen][k][\geq t]\subseteq\Ntp$ denote the set of $k$-neighborhood types which have at least $t$ occurrences in $\structgen$.

By applying the procedure presented above, we derive the following lemma:

\begin{lemma}
  \label{lem:threshold}
  Let $k,m,\delta\in\N$. There exists $T\in\N$ such that for every $\structgen\in\classe$, there exists some $t\leq T$ such that \[t\geq g(m,\delta,\sum_{\type\notin\Freq[\structgen][k][\geq t]}\nocc{\structgen}{\type})\,.\]
\end{lemma}

Let $\Freq[\structgen]:=\Freq[\structgen][k][\geq t]$ for the smallest threshold $t$ given in Lemma~\ref{lem:threshold}.

Some $k$-neighborhood type $\type\in\Ntp$ is said to be \textbf{frequent} in $\structgen\in\classe$ if it belongs to $\Freq[\structgen]$; that is, if $\nocc{\structgen}{\type}\geq t$. Otherwise, $\type$ is said to be \textbf{rare}. With the definition of $g$ in mind, Lemma~\ref{lem:threshold} can then be reformulated as follows: in every structure $\structgen\in\classe$, one can find $m$ occurrences of each frequent $k$-neighborhood type which are at distance greater than $\delta$ from one another and from the set of occurrences of every rare $k$-neighborhood type.

All that remains is for us to give a value (depending only on $k$) to the integers $m$ and $\delta$: let $M:=\max\{|\type|:\type\in\Ntp\}$ ($M$ indeed exists, and is a function of $k$ - recall that the signature $\vocab$ and the degree $d$ are assumed to be fixed). Let us consider
\begin{equation}
  \label{eq:m_and_k}
  m:=2\cdot(k+1)\cdot M!\qquad\text{and}\qquad\delta:=4k\,.
\end{equation}
We then define $\thr(k)$ as the integer $T$ provided by Lemma~\ref{lem:threshold} for these values of $m$ and $\delta$. The threshold $\thr(k)$ indeed only depends on $k$.

Finally, notice that if $\threq{\struct}{\structbis}{k}{\thr(k)}\,,$ then $\Freq=\Freq[\structbis]\,.$

As discussed in Section~\ref{sec:defneigh}, there exists a function $f$ such that $\struct\foeq{f(k)}\structbis$ entails $\threq{\struct}{\structbis}{k}{\thr(k)}$. We also make sure that $f(k)\geq\thr(k)\cdot N +1$ for every $k$.

Let us now consider $\struct,\structbis\in\classe$ such that $\struct\foeq{f(k)}\structbis$ for such an $f$.

If $\Freq=\emptyset$, then $|\struct|\leq \thr(k)\cdot N$. This guarantees that $\struct\simeq\structbis$, and in particular that $\struct\oifotwoeq{k}\structbis$. From now on, we suppose that there is at least one frequent $k$-neighborhood type.

The construction of two linear orders $\ord$ and $\ordbis$ satisfying $\orderstruct\fotwoeq{k}\orderstructbis$ is the object of Section~\ref{sec:orders}.

\subsection{\texorpdfstring{Construction of $\ord$ and $\ordbis$}{Construction of <0 and <1}}
\label{sec:orders}

This section is dedicated to the definition of two linear orders $\ord,\ordbis$ on $\struct,\structbis\in\classe$. We then prove in Section~\ref{sec:EFproof} that $\orderstruct$ and $\orderstructbis$ are \FOtwo-similar at depth $k$.

Recall that by hypothesis, $\struct$ and $\structbis$ are \FO-similar at depth $f(k)$, which entails that they have the same number of occurrences of each $\type\in\Ntp$ up to a threshold $\thr(k)$.

To construct our two linear orders, we need to define the notion of $k$-environment.

Given $\structgen\in\classe$, a linear order $\ordgen$ on $\structgendom$, $k\in\N$ and an element $\elem\in\structgendom$, we define the \textbf{$k$-environment} $\Envgen$ \textbf{of $\elem$ in $\orderstructgen$} as the restriction of $\orderstructgen$ to the $k$-neighborhood of $\elem$ in $\structgen$, where $\elem$ is the interpretation of the constant symbol $c$. Note that the order is not taken into account when determining the domain of the substructure (it would otherwise be $\structgendom$, given that any two distinct elements are adjacent for $\ordgen$).

The \textbf{$k$-environment type} $\envtpgen$ is the isomorphism class of $\Envgen$.

In other words, $\envtpgen$ contains the information of $\neighgen$ together with the order of its elements in $\orderstructgen$.

Given $\type\in\Ntp$, we define $\Ord$ as the set of $k$-environment types whose underlying $k$-neighborhood type is $\type$.

For $\indin$, we aim to partition $\structinddom$ into $2(2k+1)+2$ segments:
\[\structinddom=\Rareind\cup\bigcup_{j=0}^{2k}(\Leftind\cup\Rightind)\cup\Middleind\,.\]

Once we have set a linear order on each segment, the linear order $\ordind$ on $\structinddom$ will result from the concatenation of the orders on the segments as follows: \[(\structinddom,\ordind):=\Rareind\cdot\Leftind[0]\cdot\Leftind[1]\cdots\Leftind[2k]\cdot\Middleind\cdot\Rightind[2k]\cdots\Rightind[1]\cdot\Rightind[0]\,.\]

Each segment $\Leftind$, for $j\in\{0,\cdots,2k\}$ is itself decomposed into two segments $\NLeftind\cdot\ULeftind$. The $\ULeftind$ for $j\in\{k+1,\cdots,2k\}$ will be empty; they are defined solely in order to keep the notations uniform. The 'N' stands for ``neighbor'' and the 'U' for ``universal'', for reasons that will soon become apparent.

Symmetrically, each $\Rightind$ is decomposed into $\URightind\cdot\NRightind$, with empty $\URightind$ as soon as $j\geq k+1$.

For $\indin$ and $r\in\{0,\cdots,2k\}$, we define $\Segmentind$ as \[\Segmentind:=\Rareind\cup\bigcup_{j=0}^r(\Leftind\cup\Rightind)\,.\]

Let us now explain how the segments are constructed in $\struct$; see Figure~\ref{fig:segments} for an illustration.

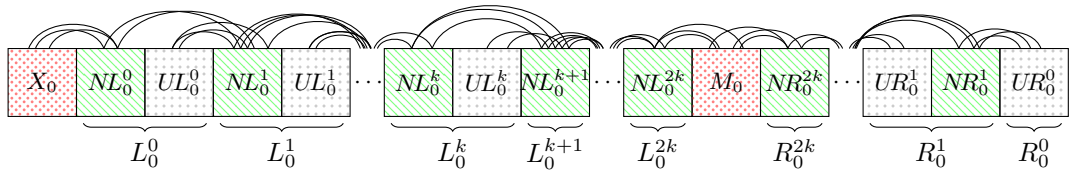
\begin{figure*}[!ht]
  \centering
  \begin{tikzpicture}[scale=.9]

    \newcommand{\topy}{.5}
    \newcommand{\boty}{.5}

    \foreach \b/\e in {1/2,3/4,5.5/6.5,7.5/8.5,9/10,11/12,13.5/14.5}
    \draw[pattern=north west lines, pattern color=green!50] (\b,-\boty) rectangle (\e,\topy);

    \foreach \b/\e in {2/3,4/5,6.5/7.5,12.5/13.5,14.5/15.5}
    \draw[pattern=dots, pattern color=black!20] (\b,-\boty) rectangle (\e,\topy);

    \foreach \b/\e in {0/1,10/11}
    \draw[pattern=crosshatch dots, pattern color=red!50] (\b,-\boty) rectangle (\e,\topy);

    \foreach \a in {5.3,8.8,12.3}
    \node at (\a,0) {\small $\cdots$};

    \node at (.5,0) {\small $\Rare$};
    \node at (1.5,0) {\small $\NLeft[0]$};
    \node at (2.5,0) {\small $\ULeft[0]$};
    \node at (3.5,0) {\small $\NLeft[1]$};
    \node at (4.5,0) {\small $\ULeft[1]$};
    \node at (6,0) {\small $\NLeft[k]$};
    \node at (7,0) {\small $\ULeft[k]$};
    \node at (8,0) {\small $\NLeft[k+1]$};
    \node at (9.5,0) {\small $\NLeft[2k]$};
    \node at (10.5,0) {\small $\Middle$};
    \node at (11.5,0) {\small $\NRight[2k]$};
    \node at (13,0) {\small $\URight[1]$};
    \node at (14,0) {\small $\NRight[1]$};
    \node at (15,0) {\small $\URight[0]$};

    \foreach \b/\e in {
      .3/1.6,.4/1.4,.6/1.6,
      1.4/3.4,1.6/3.6,
      2.4/3.3,2.4/3.8,2.5/3.5,
      3.3/5.15,3.4/5.2,3.5/5.3,3.8/5.25,
      4.4/5.3,4.4/5.31,4.5/5.22,
      5.4/5.9,5.35/6.2,
      5.9/8,5.9/8.2,6.2/7.8,
      7/7.9,7.3/8.1,
      7.8/8.65,7.9/8.6,8.1/8.7,8/8.62, 8.2/8.7,
      8.8/9.45,8.9/9.45,8.85/9.6,
      9.45/10.75,9.6/10.35, 9.45/10.4,
      10.25/11.4,10.6/11.6,10.7/11.4,
      11.4/12.15, 11.6/12.2,
      12.3/12.9,12.35/13.1,12.38/13.9,12.4/14.2,12.38/14.12,
      14.2/15.1,13.9/14.95,14.12/14.8
    }
    \draw[-] (\b,\boty) edge[out=70,in=110,-] (\e,\topy);

    \draw [decorate,decoration={brace,amplitude=3pt},xshift=0pt,yshift=0pt] (2.9,-0.6) -- (1.1,-0.6) node [black,midway,yshift=-0.4cm] {$\Left[0]$};
    \draw [decorate,decoration={brace,amplitude=3pt},xshift=0pt,yshift=0pt] (4.9,-0.6) -- (3.1,-0.6) node [black,midway,yshift=-0.4cm] {$\Left[1]$};
    \draw [decorate,decoration={brace,amplitude=3pt},xshift=0pt,yshift=0pt] (7.4,-0.6) -- (5.6,-0.6) node [black,midway,yshift=-0.4cm] {$\Left[k]$};
    \draw [decorate,decoration={brace,amplitude=3pt},xshift=0pt,yshift=0pt] (8.4,-0.6) -- (7.6,-0.6) node [black,midway,yshift=-0.4cm] {$\Left[k+1]$};
    \draw [decorate,decoration={brace,amplitude=3pt},xshift=0pt,yshift=0pt] (9.9,-0.6) -- (9.1,-0.6) node [black,midway,yshift=-0.4cm] {$\Left[2k]$};
    \draw [decorate,decoration={brace,amplitude=3pt},xshift=0pt,yshift=0pt] (11.9,-0.6) -- (11.1,-0.6) node [black,midway,yshift=-0.4cm] {$\Right[2k]$};
    \draw [decorate,decoration={brace,amplitude=3pt},xshift=0pt,yshift=0pt] (14.4,-0.6) -- (12.6,-0.6) node [black,midway,yshift=-0.4cm] {$\Right[1]$};
    \draw [decorate,decoration={brace,amplitude=3pt},xshift=0pt,yshift=0pt] (15.4,-0.6) -- (14.6,-0.6) node [black,midway,yshift=-0.4cm] {$\Right[0]$};

  \end{tikzpicture}
  \caption{The black curvy edges represent the edges between elements belonging to different segments. Edges between elements of the same segment are not represented here. The order $\ord$ grows from the left to the right.}
  \label{fig:segments}
\end{figure*}

For every $\type\in\Freq$, let $\otypefirst,\cdots,\otypelast$ be an enumeration of $\Ord$.

Recall that we defined $M$ in Section~\ref{sec:freqneigh} as $\max\{|\type|:\type\in\Ntp\}$. Thus, we have $|\Ord|\leq M!$ for every $\type\in\Ntp$.

In particular, by definition of frequency, and by choice of $m$ and $\delta$ in (\ref{eq:m_and_k}), Lemma~\ref{lem:threshold} ensures that we are able to pick, for every $\type\in\Freq$, every $l\in\{1,\cdots,|\Ord|\}$ and every $j\in\{0,\cdots,k\}$, two elements $\pinLeft$ and $\pinRight$ which have $\type$ as $k$-neighborhood type in $\struct$, such that all the $\pin$ are at distance at least $4k+1$ from each other and from any occurrence of a rare $k$-neighborhood type in $\struct$.

\paragraph*{Construction of $\Rare$ and $\NLeft[0]$}
We start with the set $\Rare$, which contains all the occurrences of rare $k$-neighborhood types, together with their $k$-neighborhoods.

Formally, the domain of $\Rare$ is $\bigcup_{\elem\in\structdom :\ \ntp\notin\Freq}\boule\,.$
\\We set $\NLeft[0]:=\Neigh{\Rare}$ (the set of neighbors of elements of $\Rare$), and define the order $\ord$ on $\Rare$ and on $\NLeft[0]$ in an arbitrary way.

\paragraph*{Construction of $\ULeft$}

If $k<j\leq 2k$, then we set $\ULeft:=\emptyset$.

Otherwise, for $j\in\{0,\cdots,k\}$, once we have constructed $\Left[0],\cdots,\Left[j-1]$ and $\NLeft$, we construct $\ULeft$ as follows.

The elements of $\ULeft$ are $\bigcup_{\tau\in\Freq}\bigcup_{l=1}^{|\Ord|}\boule[\pinLeft]\,.$

Note that $\ULeft$ does not intersect the previously constructed segments, by choice of the $\pinLeft$ and of $\delta=4k$ in (\ref{eq:m_and_k}). Furthermore, the $\boule[\pinLeft]$ are pairwise disjoint, hence we can fix $\ord$ freely and independently on each of them.

Unsurprisingly, we order each $\boule[\pinLeft]$ so that $\envtp[\pinLeft]=\otype$. This is possible because for every $\type\in\Freq$ and each $l$, $\ntp[\pinLeft]=\type$ by choice of $\pinLeft$.

Once each $\boule[\pinLeft]$ is ordered according to $\otype$, the linear order $\ord$ on $\ULeft$ can be completed in an arbitrary way.

Note that every possible $k$-environment type extending a frequent $k$-neighborhood type in $\struct$ occurs in each $\ULeft$. The $\ULeft$ are \emph{universal} in that sense.

\paragraph*{Construction of $\NLeft$}
Now, let us see how the $\NLeft$ are constructed. For $j\in\{1,\cdots,2k\}$, suppose that we have constructed $\Left[0],\cdots,\Left[j-1]$.

The domain of $\NLeft$ consists of all the neighbors (in $\struct$) of the elements of $\Left[j-1]$ not already belonging to the construction so far. Formally,
$\Neigh{\Left[j-1]}\setminus\big(\Rare\cup\bigcup_{m=0}^{j-2}\Left[m]\big).$

The order $\ord$ on $\NLeft$ is chosen arbitrarily.

\paragraph*{Construction of $\Right$}

We construct similarly the $\Right$, for $j\in\{0,\cdots,2k\}$, starting with $\NRight[0]:=\emptyset$, then $\URight[0]$ which contains each $\pinRight[\otype][0]$ together with its $k$-neighborhood in $\struct$ ordered according to $\otype$, then $\NRight[1]:=\Neigh{\Right[0]}$, then $\URight[1]$, etc.

Note that the $\pinRight$ have been chosen so that they are far enough in $\struct$ from all the segments that have been constructed so far, allowing us once more to order their $k$-neighborhood in $\struct$ as we see fit.
  
\paragraph*{Construction of $\Middle$}

$\Middle$ contains all the elements of $\structdom$ besides those already belonging to $\Segment[2k]$.
The order $\ord$ chosen on $\Middle$ is arbitrary.

\paragraph*{Transfer on $\structbis$}

Suppose that we have constructed $\Segment[2k]$.

We can make sure, retrospectively, that the index $f(k)$ in (\ref{eq:foeq}) is large enough so that there exists a set $S\subseteq\structbisdom$ so that $\struct|_{\Segment[2k]\cup\Neigh{\Segment[2k]}}\simeq\structbis|_{S}$ (this is ensured as long as $f(k)\geq|\Segment[2k]\cup\Neigh{\Segment[2k]}|+1$, which can be bounded by a function of $k$, independent of $\struct$ and $\structbis$).

Let $\bijsymb:\struct|_{\Segment[2k]}\to\structbis|_{S'}$ be the restriction to $\Segment[2k]$ of said isomorphism, and let $\bijsymb[1]$ be its converse. By construction, the $k$-neighborhood of every $\elem\in\Segment[k]$ is included in $\Segment[2k]$; hence every such $\elem$ has the same $k$-neighborhood type in $\struct$ as has $\bij[\elem]$ in $\structbis$.

We transfer alongside $\bijsymb$ all the segments, with their order, from $\orderstruct$ to $\structbis$, thus defining $\Rarebis, \NLeftbis, \ULeftbis,\cdots$ as the respective images by $\bijsymb$ of $\Rare,\NLeft,\ULeft,\cdots$, and define $\Middlebis$ as the counterpart to $\Middle$. Note that the properties concerning neighborhood are transferred; e.g. all the neighbors of an element in $\Leftbis$, $1\leq j<2k$, belong to $\Leftbis[j-1]\cup\Leftbis\cup\Leftbis[j+1]\,.$

By construction, we get the following lemma:

\begin{lemma}
  \label{lem:envpres}
  For each $\elem\in\Segment[k],$ we have $\envtp=\envtpbis[\bij[\elem]]\,.$
\end{lemma}
Lemma~\ref{lem:envpres} has two immediate consequences:
\begin{itemize}
\item The set $\Rarebis$ contains the occurrences in $\structbis$ of all the rare $k$-neighborhood types (just forget about the order on the $k$-environments, and remember that $\struct$ and $\structbis$ have the same number of occurrences of each rare $k$-neighborhood type).
\item All the universal segments $\ULeftbis$ and $\URightbis$, for $0\leq j\leq k$, contain at least one occurrence of each environment in $\Ord$, for each $\type\in\Freq$.
\end{itemize}

Our construction also guarantees the following result:

\begin{lemma}
  \label{lem:tppres}
  For each $\elem,\elembis\in\Segment[k],$ we have $\tpgen[\elem,\elembis][\orderstruct]=\tpgen[\bij[\elem],\bij[\elembis]][\orderstructbis]\,.$
\end{lemma}

In particular, for $\elem=\elembis\in\Segment[k]$, we have $\tp[\elem]=\tpbis[\bij[\elem]]\,.$

\section{\texorpdfstring{Proof of the $\FOtwo$-similarity of $\orderstruct$ and $\orderstructbis$}{Proof of the FO2-similarity of (A0,<0) and (A1,<1)}}
\label{sec:EFproof}

In this section, we aim to show the following result:

\begin{proposition}
  \label{prop:fotwoeq}
  We have that $\orderstruct\fotwoeq{k}\orderstructbis\,.$
\end{proposition}

\subsection{The two-pebble \EF game}
\label{sec:EF}

To establish Proposition~\ref{prop:fotwoeq}, we use \EF games with two pebbles. These games have been introduced by Immerman and Kozen~\cite{DBLP:journals/iandc/ImmermanK89}. Let us adapt their definition to our context.

The \textbf{$k$-round two-pebble \EF game on $\orderstruct$ and $\orderstructbis$} is played by two players: the spoiler and the duplicator. The spoiler tries to expose differences between the two structures, while the duplicator tries to establish their indistinguishability.

There are two pebbles associated with each structure: $\pebbleone$ and $\pebbletwo$ on $\orderstruct$, and $\pebbleonebis$ and $\pebbletwobis$ on $\orderstructbis$. Formally, these pebbles can be seen as the interpretations in each structure of two new constant symbols, but it will be convenient to see them as moving pieces.

At the start of the game, the duplicator places $\pebbleone$ and $\pebbletwo$ on elements of $\orderstruct$, and $\pebbleonebis$ and $\pebbletwobis$ on elements of $\orderstructbis$. The spoiler wins if the duplicator is unable to ensure that $\tp=\tpbis$. Otherwise, the proper game starts. Note that in the usual definition of the starting position, the pebbles are not on the board; however, it will be convenient to have them placed in order to uniformize our invariant. This change is not profound and does not affect the properties of the game.

For each of the $k$ rounds, the spoiler starts by choosing a structure and a pebble in this structure, and places this pebble on a element of the chosen structure. In turn, the duplicator must place the corresponding pebble in the other structure on an element of that structure. The spoiler wins at once if $\tp\neq\tpbis$. Otherwise, another round is played.

If the spoiler hasn't won after $k$ rounds, then the duplicator wins.

The main interest of these games is that they capture the expressive power of \FOtwo. We will only need the fact that these games are correct:

\begin{theorem}
  \label{th:EF}
  If the duplicator has a winning strategy in the $k$-round two-pebble \EF game on $\orderstruct$ and $\orderstructbis$, then $\orderstruct\fotwoeq{k}\orderstructbis\,.$
\end{theorem}

Thus, in order to prove Proposition~\ref{prop:fotwoeq}, we show that the duplicator wins the $r$-round two-pebble \EF game on $\orderstruct$ and $\orderstructbis\,.$

For that, let us show by a decreasing induction on $r=k,\cdots,0$ that the duplicator can ensure, after $k-r$ rounds, that the three following properties (described below) hold:
\begin{align}
  &\forall\indin,\forall\epsin,\ \pebblegen\in\Segmentind\ \to\ \pebblegen[1-\ind]=\bijgen[\pebblegen] \tag{$\labsr$}\\
  &\forall\epsin,\ \envtp[\pebblegen[0]][r]=\envtpbis[\pebblegen[1]][r] \tag{$\laber$}\\
  &\tp=\tpbis \tag{$\labtr$}
\end{align}
The first property, \refsr, guarantees that if a pebble is close (in a sense that depends on the number of rounds left in the game) to one of the $\ordind$-minimal or $\ordind$-maximal elements, the corresponding pebble in the other structure is located at the same position with respect to this $\ordind$-extremal element.

As for \refer, it states that two corresponding pebbles are always placed on elements sharing the same $r$-environment type. Once again, the satefy distance decreases with each round that goes.

Finally, \reftr controls that both pebbles have the same relative position (both with respect to the order and the original vocabulary) in the two ordered structures.

In particular, the duplicator wins the game if \reftr is satisfied at the begining of the game, and after each of the $k$ rounds of the game.

\subsection{\texorpdfstring{Base case: proofs of \refsr[k], \refer[k] and \reftr[k]}{Base case: proofs of (Sk), (Ek) and (Rk)}}

We start by proving \refsr[k], \refer[k] and \reftr[k].

At the start of the game, the duplicator places both $\pebbleone$ and $\pebbletwo$ on the $\ord$-minimal element of $\orderstruct$, and both $\pebbleonebis$ and $\pebbletwobis$ on the $\ordbis$-minimal element of $\orderstructbis$. In particular,
\[\pebbleonebis=\pebbletwobis=\bij=\bij[\pebbletwo]\,.\]

This ensures that \refsr[k] holds, while \refer[k] and \reftr[k] respectively follow from Lemma~\ref{lem:envpres} and Lemma~\ref{lem:tppres}.

\subsection{Strategy for the duplicator}
\label{subsec:strategy}

We now describe the duplicator's strategy to ensure that \refsr, \refer and \reftr hold no matter how the spoiler plays.

Suppose that we have \refsr[r+1], \refer[r+1] and \reftr[r+1] for some $0\leq r < k$, after $k-r-1$ rounds of the game.

Without loss of generality, we may assume that, in the $(k-r)$-th round of the \EF game between $\orderstruct$ and $\orderstructbis$, the spoiler moves $\pebbleone$ in $\orderstruct$.

Let us first explain informally the general idea behind the duplicator's strategy.

\begin{enumerate}
\item\label{enu:expl_border} If the spoiler plays around the endpoints (by which we mean the elements that are $\ordind$-minimal and maximal), the duplicator has no choice but to play a tit-for-tat strategy, i.e. to respond to the placement of $\pebblegen$ near the endpoints by moving $\pebblegen[1-\ind]$ on $\bijgen[\pebblegen]$.

  There are two main reasons which explain these forced moves. 

  First, if at some point, $\pebbleone$ and $\pebbleonebis$ are placed near the endpoints of their respective ordered structures (say, at respective distance $d_0,d_1$ along $\ordind$ from the $\ordind$-maximal element), but not at the same $\ordind$-distance from the endpoints (say $d_0<d_1$), then the spoiler can play in the structure $\orderstructbis$ by placing $\pebbletwobis$ on the element that is the $\ordbis$-successor of $\pebbleonebis$. The duplicator will have to place $\pebbletwo$ to the right of $\pebbleone$. Then the duplicator follows this move by placing $\pebbleonebis$ on the $\ordbis$-successor of $\pebbletwobis$, and so on. After at most $d_1$ such moves, it won't be possible for the duplicator to find a suitable element to the right of the current position of the pebbles in $\orderstruct$. Of course, the duplicator must not only maintain $\ordtp=\ordtpbis$ at all time, but also make sure that $\vocabtp=\vocabtpbis$. Hence, a small right shift by the spoiler in $\orderstructbis$ may force the duplicator to make a big right shift in $\orderstruct$, and it might take less than $d_1$ rounds for the spoiler to win the game in this situation; in other words, the borders around the endpoints in which the duplicator must apply the tit-for-tat strategy are larger than just the number of rounds in the game.

  Second, the occurrences of rare neighborhood types are located in $\orderstructind$ near the $\ordind$-minimal element. If the duplicator does not play according to $\bijsymb$ in this area, it will be easy enough for the spoiler to win the game.

  The reason we introduced the segments $\NLeftind,\ULeftind,\NRightind$ and $\URightind$ is precisely to bound the area in which the duplicator must implement the tit-for-tat strategy. Indeed, as soon as a pebble is placed in $\Middleind$, there is no way for the spoiler to join the endpoints in less than $k$ moves while forcing the duplicator's hand.

  The case where the spoiler plays near the endpoints corresponds to Case~(\ref{enu:border}) below, and is detailed in Section~\ref{sec:border}.

\item\label{enu:expl_neigh} Next, suppose that the spoiler places a pebble, say $\pebbleone$, next (in $\struct$) to $\pebbletwo$, i.e. such that $\pebbleone\in\voisdom$. The duplicator must place $\pebbleonebis$ on an element whose relative position to $\pebbletwobis$ is the same as the relative position of $\pebbleone$ with respect to $\pebbletwo$. Note that once this is done, the spoiler can change variable, and place $\pebbletwo$ (or $\pebbletwobis$, if they decide to play in $\orderstructbis$) in $\voisdom[\pebbleone]$, thus forcing the duplicator to play near $\pebbleonebis$. In order to prevent the spoiler from being able, in $k$ such moves, to expose the difference between $\orderstruct$ and $\orderstructbis$, the duplicator must make sure, with $r$ rounds left, that $\pebbleone$ and $\pebbleonebis$ (as well as $\pebbletwo$ and $\pebbletwobis$) share the same $r$-environment in $\orderstruct$ and $\orderstructbis$. This will guarantee that the duplicator can play along if the spoiler decides to take $r$ moves adjacent (in $\structind$) to one another.

  The case where the spoiler places a pebble next (in the structure without ordering) to the other pebble is our Case~(\ref{enu:neigh}), and is treated in Section~\ref{sec:neigh}.

\item\label{enu:expl_left} Suppose now that the spoiler's move does not fall under the previous templates. Let us assume that the spoiler plays in $\orderstruct$, and moves $\pebbleone$ to the left of $\pebbletwo$ (i.e. such that $\orderstruct\models\pebbleone<\pebbletwo$).

  In order to play according to the remarks from Cases~\ref{enu:expl_border} and \ref{enu:expl_neigh}, the duplicator must place $\pebbleonebis$ on an element which shares the same $r$-environment with $\pebbleone$ (where $r$ is the number of rounds left in the game), which is not near the endpoints.

  It must be the case that the $k$-neighborhood type of $\pebbleone$ in $\struct$ is frequent, since it is not near the endpoints of $\orderstruct$, hence not in $\Rare$. By construction, every universal segment $\ULeftbis$, for $0\leq j\leq k$, contains elements of each $k$-environment type extending any frequent $k$-neighborhood type. In particular, it contains an element having the same $r$-environment as $\pebbleone$. The duplicator will place $\pebbleonebis$ on such an element in the leftmost segment $\ULeftbis$ which is not considered to be near the endpoints (this notion depends on the number $r$ of rounds left in the game). This is detailed in Cases~(\ref{enu:farleft}) and (\ref{enu:farright}) (for the symmetrical case where $\pebbleone$ is placed to the right of $\pebbletwo$) below.

  However, we have to consider a subcase, where $\pebbletwobis$ is itself in the leftmost segment $\Leftbis$ which is not near the endpoints. Indeed, in this case, placing $\pebbleonebis$ as discussed may result in $\pebbleonebis$ being to the right of $\pebbletwobis$, or being in $\voisdombis$; either of which being game-losing to the duplicator. However, since $\pebbletwobis$ was considered to be near the endpoints in the previous round of the game, we know that the duplicator played a tit-for-tat strategy at that point, which allows us to replicate the placement of $\pebbleone$ according to $\bijsymb$. This subcase, as well as the equivalent subcase where the spoiler places $\pebbleone$ to the right of $\pebbletwo$, are formalized in Cases~(\ref{enu:left}) and (\ref{enu:right}) below.

\end{enumerate}

We are now ready to describe formally the strategy implemented by the duplicator:

\begin{enumerate}[(I)]
\item\label{enu:border} If $\pebbleone\in\Segment$, then the duplicator responds by placing $\pebbleonebis$ on $\bij$.

  This corresponds to the tit-for-tat strategy implemented when the spoiler plays near the endpoints, as discussed in Case~\ref{enu:expl_border}.

\item\label{enu:neigh} Else, if $\pebbleone\notin\Segment$, and $\pebbleone\in\voisdom$, then \refer[r+1] ensures that there exists an isomorphism $\psi:\Env[\pebbletwo][r+1]\to\Envbis[\pebbletwobis][r+1]\,.$ The duplicator responds by placing $\pebbleonebis$ on $\psi(\pebbleone)$.

  This makes formal the duplicator's response to a move next to the other pebble, as discussed in Case~\ref{enu:expl_neigh} above.

\item\label{enu:farleft} Else suppose that $\orderstruct\models\pebbleone<\pebbletwo$ and $\pebbletwo\notin\Left[r+1]$. Note that $\type:=\ntp[\pebbleone]\in\Freq$, since $\pebbleone\notin\Rare$. Let $\otype:=\envtp[\pebbleone]$.
  
  The duplicator responds by placing $\pebbleonebis$ on $\bij[\pinLeft[\otype][r+1]]$.

\item\label{enu:left} Else, if $\orderstruct\models\pebbleone<\pebbletwo$ and $\pebbletwo\in\Left[r+1]$, then the duplicator moves $\pebbleonebis$ on $\bij$ (by \refsr[r+1], $\pebbleone$ indeed belongs to the domain of $\bijsymb$).

\item\label{enu:farright} Else, suppose that $\orderstruct\models\pebbletwo<\pebbleone$ and $\pebbletwo\notin\Right[r+1]$. This case is symmetric to Case~(\ref{enu:farleft}).

  Similarly, the duplicator opts to play $\pebbleonebis$ on $\bij[\pinRight[\otype][r+1]]$, where $\otype:=\envtp[\pebbleone]$.

\item\label{enu:right} If we are in none of the cases above, it means that the spoiler has placed $\pebbleone$ to the right of $\pebbletwo$, and that $\pebbletwo\in\Right[r+1]$. This case is symmetric to Case~(\ref{enu:left}).

  Once again, the duplicator places $\pebbleonebis$ on $\bij$.
\end{enumerate}

Let us now prove that this strategy indeed satisfies our invariants: under the inductive assumption that \refsr[r+1], \refer[r+1] and \reftr[r+1] hold, for some $0\leq r<k$, we show that this strategy ensures that \refsr, \refer and \reftr hold.

We treat each case in its own section: Section~\ref{sec:border} is devoted to Case~(\ref{enu:border}) while Section~\ref{sec:neigh} covers Case~(\ref{enu:neigh}). Both Cases~(\ref{enu:farleft}) and (\ref{enu:left}) are treated in Section~\ref{sec:left}. Cases~(\ref{enu:farright}) and (\ref{enu:right}), being their exact symmetric counterparts, are left to the reader.

\begin{note}
  \label{note:y}
Note that some properties need no verification. Since $\pebbletwo$ and $\pebbletwobis$ are left untouched by the players, \refsr[r+1] ensures that half of \refsr automatically holds, namely that
\[\forall\indin,\quad\pebblegen[\ind][y]\in\Segmentind\quad\to\quad\pebblegen[1-\ind][y]=\bijgen[\pebblegen[\ind][y]]\,.\]

Similarly, the part of \refer concerning $\pebbletwo$ and $\pebbletwobis$ follows from \refer[r+1]:
\[\envtp[\pebbletwo][r]=\envtpbis[\pebbletwobis][r]\,.\]

Lastly, notice that once we have shown that \refer holds, it follows that
\[
\left\{
\begin{array}{l}
  \vocabtpone=\vocabtponebis\\
  \vocabtptwo=\vocabtptwobis
\end{array}
\right.
\]


\end{note}

\subsection{When the spoiler plays near the endpoints: Case~(\ref{enu:border})}
\label{sec:border}

In this section, we treat the case where the spoiler places $\pebbleone$ near the $\ord$-minimal or $\ord$-maximal element of $\orderstruct$. Obviously, what ``near'' means depends on the number of rounds left in the game; the more rounds remain, the more the duplicator must be cautious regarding the possibility for the spoiler to reach an endpoint and potentially expose a difference between $\orderstruct$ and $\orderstructbis$.

As we have stated in Case~(\ref{enu:border}), with $r$ rounds left, we consider a move on $\pebbleone$ by the spoiler to be near the endpoints if it is made in $\Segment$. In that case, the duplicator responds along the tit-for-tat strategy, namely by placing $\pebbleonebis$ on $\bij$.

Let us now prove that this strategy guarantees that \refsr, \refer and \reftr hold. Recall from Note~\ref{note:y} that part of the task is already taken care of.

\paragraph*{Proof of \refsr in Case~(\ref{enu:border})}

We have to show that $\forall\indin,\ \pebblegen[\ind][x]\in\Segmentind\ \to\ \pebblegen[1-\ind][x]=\bijgen[\pebblegen[\ind][x]]\,.$
This follows directly from the duplicator's strategy, since $\pebbleonebis=\bij$ (thus $\pebbleone=\bijbis$).

\paragraph*{Proof of \refer in Case~(\ref{enu:border})}

We need to prove that $\envtp[\pebbleone][r]=\envtpbis[\pebbleonebis][r]\,,$ which is a consequence of Lemma~\ref{lem:envpres} given that $\pebbleonebis=\bij$ and $r<k$.

\paragraph*{Proof of \reftr in Case~(\ref{enu:border})}

First, suppose that $\pebbletwo\in\Segment[r+1]$. By \refsr[r+1], we know that $\pebbletwobis=\bij[\pebbletwo]$. Thus, Lemma~\ref{lem:tppres} allows us to conclude that $\tp=\tpbis$.

Otherwise, $\pebbletwo\notin\Segment[r+1]$ and \refsr[r+1] entails that $\pebbletwobis\notin\Segmentbis[r+1]$.

We have two points to establish:
\begin{align}
  \vocabtp=\vocabtpbis\label{eq:I_vocab}\\
  \ordtp=\ordtpbis\label{eq:I_order}
\end{align}

Notice that
\[
\left\{
\begin{array}{l}
  \vocabtp=\vocabtpone\cup\vocabtptwo\\
  \vocabtpbis=\vocabtponebis\cup\vocabtptwobis
\end{array}
\right.
\]
This is because, by construction, the neighbors in $\structind$ of an element of $\Segmentind$ all belong to $\Segmentind[r+1]$. Equation~(\ref{eq:I_vocab}) follows from this remark and Note~\ref{note:y}.

As for Equation~(\ref{eq:I_order}), either \[\pebbleone\in\Rare\cup\bigcup_{0\leq j\leq r}\Left\quad\text{and}\quad\pebbleonebis\in\Rarebis\cup\bigcup_{0\leq j\leq r}\Leftbis\,,\] in which case $\ordtp=\{x<y\}=\ordtpbis\,,$ or \[\pebbleone\in\bigcup_{0\leq j\leq r}\Right\quad\text{and}\quad\pebbleonebis\in\bigcup_{0\leq j\leq r}\Rightbis\,,\] in which case $\ordtp=\{x>y\}=\ordtpbis\,.$

\subsection{When the spoiler plays next to the other pebble: Case (\ref{enu:neigh})}
\label{sec:neigh}

Suppose now that the spoiler places $\pebbleone$ next to the other pebble in $\struct$ (i.e. $\pebbleone\in\voisdom$), but not in $\Segment$ (for that move would fall under the jurisdiction of Case~(\ref{enu:border})).

In that case, the duplicator must place $\pebbleonebis$ so that the relative position of $\pebbleonebis$ and $\pebbletwobis$ is the same as that of $\pebbleone$ and $\pebbletwo$.

For that, we can use \refer[r+1], which guarantees that $\envtp[\pebbletwo][r+1]=\envtpbis[\pebbletwobis][r+1]\,.$ Thus there exists an isomorphism $\psi$ between $\Env[\pebbletwo][r+1]$ and $\Envbis[\pebbletwobis][r+1]$. Note that this isomorphism is unique, by virtue of $\ord$ and $\ordbis$ being linear orders.

The duplicator's response is to place $\pebbleonebis$ on $\psi(\pebbleone)$. Let us now prove that this strategy is correct with respect to our invariants \refsr, \refer and \reftr.

\paragraph*{Proof of \refsr in Case~(\ref{enu:neigh})}

Because the spoiler's move does not fall under Case~(\ref{enu:border}), we know that $\pebbleone\notin\Segment$.

Let us now show that $\pebbleonebis$ is not near the endpoints either: suppose that $\pebbleonebis\in\Segmentbis$. By construction, since $\pebbleonebis$ and $\pebbletwobis$ are neighbors in $\structbis$, this entails that $\pebbletwobis\in\Segmentbis[r+1]$. But then, we know by \refsr[r+1] that $\pebbletwo=\bijbis[\pebbletwobis]$; and because $\psi$ is the unique isomorphism between $\Env[\pebbletwo][r+1]$ and $\Envbis[\pebbletwobis][r+1]$, $\psi$ is equal to the restriction $\widetilde\bijsymb$ of $\bijsymb$:
\[\widetilde\bijsymb\ :\ \Env[\pebbletwo][r+1]\ \to\ \Envbis[\pebbletwobis][r+1]\,.\]

Thus $\pebbleone=\psi^{-1}(\pebbleonebis)=\widetilde\bijsymb^{-1}(\pebbleonebis)=\bijbis$, and by definition of the segments on $\orderstructbis$, which are just a transposition of the segments of $\orderstruct$ via $\bijsymb$, $\pebbleonebis\in\Segmentbis$ then entails that $\pebbleone\in\Segment$, which is absurd.

Since we neither have $\pebbleone\in\Segment$ nor $\pebbleonebis\in\Segmentbis$, \refsr holds - recall from Note~\ref{note:y} that the part concerning $\pebbletwo$ and $\pebbletwobis$ is always satisfied.

\paragraph*{Proof of \refer in Case~(\ref{enu:neigh})}

Recall that the duplicator placed $\pebbleonebis$ on the image of $\pebbleone$ by the isomorphism \[\psi:\Env[\pebbletwo][r+1]\to\Envbis[\pebbletwobis][r+1]\,.\]

It is easy to check that the restriction $\widetilde\psi$ of $\psi$:
$\widetilde\psi\ :\ \Env[\pebbleone][r]\ \to\ \Envbis[\pebbleonebis][r]$
is well defined, and is indeed an isomorphism.\\ This ensures that $\envtp[\pebbleone][r]=\envtpbis[\pebbleonebis][r]\,,$ thus completing the proof of \refer.

\paragraph*{Proof of \reftr in Case~(\ref{enu:neigh})}

This follows immediately from the fact that the isomorphism $\psi$ maps $\pebbleone$ to $\pebbleonebis$ and $\pebbletwo$ to $\pebbletwobis$: all the atomic facts about these elements are preserved.

\subsection{When the spoiler plays to the left: Cases~(\ref{enu:farleft}) and (\ref{enu:left})}
\label{sec:left}

We now treat our last case, which covers both Cases~(\ref{enu:farleft}) and (\ref{enu:left}), i.e. the instances where the spoiler places $\pebbleone$ to the left of $\pebbletwo$ (formally: such that $\orderstruct\models\pebbleone<\pebbletwo$), which do not already fall in Cases~(\ref{enu:border}) and (\ref{enu:neigh}).

Note that the scenario in which the spoiler plays to the right of the other pebble is the exact symmetric of this one (since the $\Rareind$ play no role in this case, left and right can be interchanged harmlessly).

The idea here is very simple: since the spoiler has placed $\pebbleone$ to the left of $\pebbletwo$, but neither in $\Segment$ nor in $\voisdom$, the duplicator responds by placing $\pebbleonebis$ on an element of $\ULeftbis[r+1]$ (the leftmost universal segment not in $\Segmentbis$) sharing the same $k$-environment.

This is possible by construction of the universal segments: if $\otype:=\envtp[\pebbleone]$ (which must extend a frequent $k$-neighborhood type, since $\pebbleone\notin\Rare$), then $\bij[\pinLeft[\otype][r+1]]$ satisfies the requirements.

There is one caveat to this strategy. If $\pebbletwobis$ is itself in $\Leftbis[r+1]$, two problems may arise: first, it is possible for $\pebbleonebis$ and $\pebbletwobis$ to be in the wrong order (i.e. such that $\orderstructbis\models\pebbleonebis>\pebbletwobis$). Second, it may be the case that $\pebbleonebis$ and $\pebbletwobis$ are neighbors in $\structbis$, which, together with the fact that $\pebbleone$ and $\pebbletwo$ are orthogonal in $\struct$ (i.e. $\vocabtp=\vocabtpone\cup\vocabtptwo$), would break \reftr.

This is why the duplicator's strategy depends on whether $\pebbletwobis\in\Leftbis[r+1]$:

\begin{itemize}
\item if this is not the case, then the duplicator places $\pebbleonebis$ on $\bij[\pinLeft[\otype][r+1]]$. This corresponds to Case~(\ref{enu:farleft}).
\item if $\pebbletwobis\in\Leftbis[r+1]$, then \refsr[r+1] guarantees that $\pebbletwo\in\Left[r+1]$. Hence $\pebbleone$, which is located to the left of $\pebbletwo$, is in the domain of $\bijsymb$: the duplicator moves $\pebbleonebis$ to $\bij$. This situation corresponds to Case~(\ref{enu:left}).
\end{itemize}

Let us prove that \refsr, \refer and \reftr hold in both of these instances.

\paragraph*{Proof of \refsr in Case~(\ref{enu:farleft})}

Since the spoiler's move does not fall under Case~(\ref{enu:border}), we have that $\pebbleone\notin\Segment$.

By construction, $\pinLeft[\otype][r+1]\in\Left[r+1]$, thus $\bij[\pinLeft[\otype][r+1]]\in\Leftbis[r+1]$, and $\pebbleonebis\notin\Segmentbis$.

\paragraph*{Proof of \refer in Case~(\ref{enu:farleft})}

It follows from $\envtp[\pinLeft[\otype][r+1]]=\otype$ together with Lemma~\ref{lem:envpres} that \[\envtp[\pebbleone]=\envtpbis[\pebbleonebis]\,.\]
A fortiori, $\envtp[\pebbleone][r]=\envtpbis[\pebbleonebis][r]$.

\paragraph*{Proof of \reftr in Case~(\ref{enu:farleft})}

Because the spoiler's move does not fall under Case~(\ref{enu:neigh}), $\pebbleone\notin\voisdom$. In other words, \[\vocabtp=\vocabtpone\cup\vocabtptwo\,.\]
Recall the construction of $\ULeft[r+1]$: the whole $k$-neighborhood of $\pinLeft[\otype][r+1]$ was included in this segment. In particular, $\boule[\pebbleonebis][\structbis][1]=\boule[\bij[\pinLeft[\otype][r+1]]][\structbis][1]\subseteq\ULeftbis[r+1]\,.$
By assumption, $\pebbletwobis\notin\Leftbis[r+1]$, which entails that $\vocabtpbis=\vocabtponebis\cup\vocabtptwobis\,.$
\\
It then follows from the last observation of Note~\ref{note:y} that $\vocabtp=\vocabtpbis\,.$

Let us now prove that $\ordtpbis=\{x<y\}$.
\\
We claim that $\pebbletwobis\notin\Rarebis\cup\bigcup_{0\leq j\leq r+1}\Leftbis$. Suppose otherwise: \refsr[r+1] would entail that $\pebbletwo\in\Rare\cup\bigcup_{0\leq j\leq r+1}\Left$ which, together with the hypothesis $\pebbletwo\notin\Left[r+1]$ and $\pebbleone<\pebbletwo$, would result in $\pebbleone$ being in $\Segment$, which is absurd.
\\
Thus, $\ordtpbis=\{x<y\}=\ordtp$, which concludes the proof of \reftr.

\paragraph*{Proof of \refsr, \refer and \reftr in Case~(\ref{enu:left})}

Let us now move to the case where $\pebbletwobis\in\Leftbis[r+1]$. Recall that under this assumption, $\pebbletwo=\bijbis[\pebbletwobis]\in\Left[r+1]$ and since $\pebbleone<\pebbletwo$ and $\pebbleone\notin\Segment$, we have that $\pebbleone\in\Left[r+1]$.

The duplicator places the pebble $\pebbleonebis$ on $\bij$; in particular, $\pebbleonebis\in\Leftbis[r+1]$.

The proof of \refsr follows from the simple observation that $\pebbleone\notin\Segment$ and $\pebbleonebis\notin\Segmentbis$.

As for \refer and \reftr, they follow readily from Lemma~\ref{lem:envpres} and \ref{lem:tppres} and the fact that
$\pebbleonebis=\bij$ and $\pebbletwobis=\bij[\pebbletwo]$.

\section{Counting quantifiers}
\label{sec:counting}

We now consider the extension $\Ctwo$ of $\FOtwo$, where one is allowed to use counting quantifiers of the form $\existsm{i}x$ and $\existsm{i}y$, for $i\in\N$. Such a quantifier, as expected, expresses the existence of at least $i$ elements satisfying the formula which follows it.

This logic $\Ctwo$ has been extensively studied. On an expressiveness standpoint, $\Ctwo$ stricly extends $\FOtwo$ (which cannot count up to three), and contrary to the latter, \Ctwo does not enjoy the small model property. However, the satisfiability problem for $\Ctwo$ is still decidable~\cite{DBLP:conf/lics/GradelOR97,DBLP:journals/logcom/Pratt-Hartmann07,DBLP:conf/wollic/Pratt-Hartmann10}. To the best of our knowledge, it is not known whether $\oictwo$ has a decidable syntax.

Let us now explain how the proof of Theorem~\ref{th:main} can be adapted to show the following stronger version:

\begin{theorem}
  \label{th:mainCtwo}
  Let $\classe$ be a class of structures of bounded degree.

  Then $\oictwo\subseteq\FO$ on $\classe$.
\end{theorem}

\begin{proof}
  The proof is very similar as to that of Theorem~\ref{th:main}. The difference is that we now need to show, at the end of the construction, that the structures $\orderstruct$ and $\orderstructbis$ are not only $\FOtwo$-similar, but $\Ctwo$-similar. More precisely, we show that for every $k\in\N$, there exists some $f(k)\in\N$ such that if $\struct\foeq{f(k)}\structbis$, then it is possible to construct two linear orders $\ord$ and $\ordbis$ such that $\orderstruct$ and $\orderstructbis$ agree on all $\Ctwo$-sentences of quantifier rank at most $k$, and with counting indexes at most $k$, which we denote $\orderstruct\ctwoeq{k,k}\orderstructbis\,.$
  
  This is enough to complete the proof, as these classes of $\Ctwo$-sentences cover all the $\Ctwo$-definable properties.

  In order to prove that $\orderstruct\ctwoeq{k,k}\orderstructbis$, we need an \EF-game capturing $\ctwoeq{k,k}$. It is not hard to derive such a game from the \EF-game for $\Ctwo$~\cite{immerman1990describing}.

  This game only differs from the two-pebble \EF-game in that in each round, once the spoiler has chosen a structure (say $\orderstruct$) and a pebble to move (say $\pebbleone$), the spoiler picks not only one element of that structure, but a set $\set$ of up to $k$ elements. Then the duplicator must respond with a set $\setbis$ of same cardinality in $\orderstructbis$. The spoiler then places $\pebbleonebis$ on any element of $\setbis$, to which the duplicator responds by placing $\pebbleone$ on some element of $\set$. As usual, the spoiler wins after this round if $\tp\neq\tpbis\,.$ Otherwise, the games goes on until $k$ rounds are played.

  It is not hard to establish that this game indeed captures $\ctwoeq{k,k}$, in the sense that $\orderstruct\ctwoeq{k,k}\orderstructbis$ if and only if the duplicator has a winning strategy for $k$ rounds of this game. The restriction on the cardinal of the set chosen by the spoiler (which is at most $k$) indeed corresponds to the fact that the counting indexes of the formulas are at most $k$. As for the number of rounds (namely, $k$), it corresponds as usual to the quantifier rank. This can be easily derived from a proof of Theorem 5.3 in~\cite{immerman1990describing}, and is left to the reader.

  Let us now explain how to modify the construction of $\ord$ and $\ordbis$ presented in Section~\ref{sec:construction} in order for the duplicator to maintain similarity for $k$-round in such a game. The only difference lies in the choice of the universal elements. Recall that in the previous construction, we chose, for each $k$-environment type $\tau_l$ extending a frequent $k$-neighborhood type and each segment $\ULeft$, an element $\pinLeft$ whose $k$-environment type in $\orderstruct$ is destined to be $\tau_l$ (and similarly for $\URight$ and $\pinRight$).

  In the new construction, we pick $k$ such elements, instead of just one. Just as previously, all these elements must be far enough from one another in the Gaifman graph of $\struct$. Once again, this condition can be met by virtue of the $k$-neighborhood type $\tau$ underlying $\tau_l$ being frequent, and thus having many occurrences scattered across $\struct$ (remember that we have a bound on the degree of $\struct$, thus all the occurrences of $\tau$ cannot be concentrated). We only need to multiply the value of $m$ by $k$ in~(\ref{eq:m_and_k}).
  
  When the spoiler picks a set of elements of size at most $k$ in one of the structures (say $\set$ in $\orderstruct$), the duplicator responds by selecting, for each one of the elements of $\set$, an element in $\orderstructbis$ along the strategy for the $\FOtwo$-game explained in Section~\ref{subsec:strategy}. All that remains to be shown is that it is possible for the duplicator to answer each element of $\set$ with a different element in $\orderstructbis$.

  Note that if the duplicator follows the strategy from Section~\ref{subsec:strategy}, they will never answer two moves by the spoiler falling under different cases among Cases~(\ref{enu:border})-(\ref{enu:right}) with the same element. Thus we can treat separately each one of these cases; and for each case, we show that if the spoiler chooses up to $k$ elements in $\orderstruct$ falling under this case in $\set$, then the duplicator can find the same number of elements in $\orderstructbis$, following the aforementionned strategy.
  \begin{itemize}
  \item For Case~(\ref{enu:border}), this is straightforward, since the strategy is based on the isomorphism between the borders of the linear orders. The same goes for Cases~(\ref{enu:neigh}), (\ref{enu:left}) and (\ref{enu:right}), as the strategy in these cases also relies on an isomorphism argument.
  \item Suppose now that $\pebbletwo\notin\Left[r+1]$, and assume that the spoiler chooses several elements to the left of $\pebbletwo$, but outside of $\Segment$ and not adjacent to $\pebbletwo$. This corresponds to Case~(\ref{enu:farleft}). Recall that our new construction guarantees, for each $k$-environment type extending a frequent $k$-neighborhood type, the existence in $\Leftbis[r+1]$ of $k$ elements having this environment. This lets us choose, in $\Leftbis[r+1]$, a distinct answer for each element in the set selected by the spoiler, sharing the same $k$-environment type. Case~(\ref{enu:farright}) is obviously symmetric.
  \end{itemize}
  This concludes the proof of Theorem~\ref{th:mainCtwo}.
\end{proof}

\section{Conclusion}

\newcommand{\horizoffset}{16}

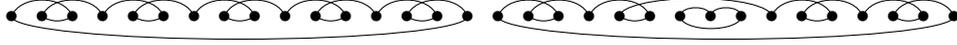
\begin{figure*}[!ht]
  \centering
  \begin{tikzpicture}[scale=.4]

    \foreach \a in {0,...,15}
    \node at (\a,0) {$\bullet$};

    \foreach \a/\b in {1/2,4/5,7/8,10/11,13/14}
    \draw[-] (\a,0) edge[bend right] (\b,0);

    \foreach \a/\b in {0/2,1/3,3/5,4/6,6/8,7/9,9/11,10/12,12/14,13/15}
    \draw[-] (\a,0) edge[bend left=70] (\b,0);

    \draw (0,0) arc (180:360:7.5cm and .75cm);

    \foreach \a in {0,...,15}
    \node at ({{\a+\horizoffset}},0) {$\bullet$};

    \foreach \a/\b in {1/2,4/5,10/11,13/14}
    \draw[-] ({{\a+\horizoffset}},0) edge[bend right] ({{\b+\horizoffset}},0);

    \foreach \a/\b in {0/2,1/3,3/5,9/11,10/12,12/14,13/15}
    \draw[-] ({{\a+\horizoffset}},0) edge[bend left=70] ({{\b+\horizoffset}},0);

    \foreach \a/\b in {6/7,7/8}
    \draw[-] ({{\a+\horizoffset}},0) edge[bend left=90] ({{\b+\horizoffset}},0);

    \draw ({{6+\horizoffset}},0) arc (180:360:1cm and .4cm);

    \draw ({{4+\horizoffset}},0) arc (180:0:2.5cm and .6cm);
        
    \draw (\horizoffset,0) arc (180:360:7.5cm and .75cm);

  \end{tikzpicture}
  \caption{Illustration of two linear orders (growing from left to right) on a cycle (left figure) and the disjoint union of a cycle and a triangle (right figure), which are indistinguishable by any \Ctwo-sentence of small enough quantifier rank and maximal counting index (where ``small'' is understood with respect to the length of the cycles).}
  \label{fig:sep}
\end{figure*}

We have established that, when the degree is bounded, properties definable in the order-invariant extension of the two-variable fragment of first-order logic with counting are definable in first-order logic.

From there, there seem to be three axes in which one can try to complete the picture.

First, a natural question is whether this inclusion of expressiveness still holds when we release the hypothesis on the degree. The classical examples separating \oifo from \FO seem to require at least three variables (for example, the Potthoff separating example~\cite{potthoff1994logische} can be expressed with three variables, thus proving that $\oifothree\not\subseteq\FO$). It would be very interesting to know whether such a separating example exists for \oictwo, or even for \oifotwo.

Second, it is quite clear that the inclusion $\oictwo\subseteq\FO$ when the degree is bounded is a severe over-approximation of the expressive power of \oictwo. For instance, it is not hard to prove that \oictwo cannot define the class of triangle-free graphs: no sentence from \oictwo can make a distinction between a large enough cycle and the disjoint union of a large enough cycle together with a triangle. This can be seen, following the general strategy detailed in Section~\ref{sec:main}, by constructing two carefully chosen linear orders on these graphs. Figure~\ref{fig:sep} illustrates the construction of such orders. Notice that there are only three kinds of elements: those whose two neighbors are on their left, those for which they are on their right, and those which have one neighbor on each side. By making sure to always respond to a move by the spoiler with an element of the same kind (and, of course, by implementing a tit-for-tat strategy near the endpoints), the duplicator can easily win the \EF-game capturing $\ctwoeq{k,k}$, provided that the cycles are long enough with respect to $k$.

It would be interesting to find upper bounds for \oifotwo and \oictwo tighter than \FO - that is, tighter than the fragment $\exists^*\forall^*\exists^*\FO$, to which \FO collaspes when the degree is bounded~\cite{DBLP:journals/iandc/FaginSV95} (since already in this fragment, one can count the number of occurrences of neighborhood types up to some threshold). Let us briefly explain why we fall short of giving such a bound: in such an attempt, the initial assumption about the similarity between $\struct$ and $\structbis$ would be weaker than \FO-similarity, and it would not be possible to base our work on neighborhoods. In this context, the starting hypothesis on the two structures would lack the rigidity which seems necessary to construct linear orders preserving \FOtwo- or \Ctwo-similarity. Establishing such a tighter bound thus seems to call for new techniques.

Last, we conjecture that the inclusion still holds when we lift the restriction on the number of variables; namely that $\oifo=\FO$ when the degree is bounded. This would generalize the equality $\sifo=\FO$ when the degree is bounded, obtained in~\cite{DBLP:conf/lics/Grange20}. Our hope is that a construction inspired by this one - albeit significantly refined - and in particular by the alternation of universal and neighbors segments, could possibly lead to establish such a result.

We leave these three questions, as well as the issue of the syntactic decidability of \oictwo, for further research.

\bibliography{biblio}

\end{document}